\documentclass[onecolumn]{IEEEtran}

\usepackage{xcolor,graphicx}
\usepackage[caption=false,font=footnotesize]{subfig}
\usepackage{amsmath,amsthm}
\usepackage{amsfonts}
\usepackage{amssymb}
\usepackage{bm}
\usepackage[normalem]{ulem}
\usepackage{comment}
\usepackage{cite}
\newcommand{\SaS}[1][\alpha]{S#1S}

\newcommand{\snr}{\mathrm{snr}_{\alpha}}

\DeclareMathOperator*{\argmin}{arg\,min}

\newtheorem{definition}{Definition}
\newtheorem{lemma}{Lemma}

\newtheorem{corollary}{Corollary}
\newtheorem{remark}{Remark}
\newtheorem{property}{Property}
\newtheorem{proposition}{Proposition}
\newtheorem{theorem}{Theorem}

\def\BibTeX{{\rm B\kern-.05em{\sc i\kern-.025em b}\kern-.08em
    T\kern-.1667em\lower.7ex\hbox{E}\kern-.125emX}}

\begin{document}
\title{On the Information Bottleneck for Stable Variables} 

% %%% Single author, or several authors with same affiliation:
% \author{%
%  \IEEEauthorblockN{Author 1 and Author 2}
% \IEEEauthorblockA{Department of Statistics and Data Science\\
%                    University 1\\
 %                   City 1\\
  %                  Email: author1@university1.edu}% }

\author{
        Dongmin~Park,~Jihad~Fahs and Mohamad Assaad
        
\thanks{ 
 Dongmin Park and Jihad Fahs are with the Department of Electrical and Computer Engineering at the American University of Beirut (AUB), Beirut, Lebanon (emails: \{djp00,jihad.fahs\}@aub.edu.lb). Jihad Fahs is a Fellow of the Center for Advanced Mathematical Sciences at the AUB. Mohamad Assaad is with the Laboratoire des Signaux et Syst\`emes (L2S) at CentraleSupelec, CNRS, University of Paris-Saclay, Paris, France (email: Mohamad.Assaad@centralesupelec.fr). The research of M. Assaad was funded in part by the Agence Nationale
de la Recherche (ANR) under the ANR-24-IAS1-0002-02
COMSEMA project. }
}

\maketitle

\begin{abstract}
We consider a generalization of the Gaussian information bottleneck problem to jointly stable variables $X$ and $Y$. Specifically, for $X = Y + A$, where $A$ is stable and independent of $Y$, we find optimal stable linear encoders of the form $T = kX + N$, with $N$ stable and independent of $X$. We characterize the linear stable information bottleneck in closed-form and its critical value. Furthermore, we show that, except for the Gaussian case, such stochastic linear encoders are strictly suboptimal, while becoming asymptotically optimal at large compression rates. Our linear solution recovers the Gaussian information bottleneck for jointly Gaussian variables as a special case.   
\end{abstract}

 %Information bottleneck, alpha-stable distributions, heavy-tailed distributions, linear representations, independent component analysis, dimensionality reduction.

\section{Introduction}
The information bottleneck (IB) framework provides a principled approach to learning representations that preserve information about a relevant variable while compressing the input data. This framework has been applied to problems such as semantic communication and deep learning~\cite{e26020102,7133169,10438074,Goldfeld2020}. Of particular relevance to this work, IB can be viewed as a remote lossy source coding problem under a logarithmic loss \cite{Dobrushin1962,Witsenhausen1980,Zaidi2020}. More explicitly, consider a variable \(X\) that is statistically dependent on another variable of interest \(Y\). The goal is to find an optimal stochastic representation or encoder \(T\) for input data \(X\), one that minimizes the compression rate while preserving a certain level of information about \(Y\). Formally, both notions of compression and information relevance are expressed through mutual information, and finding the optimal representation \(T\) (more specifically, the conditional probability law \(p_{T|X}\)) is done through the following optimization problem:
\begin{equation}\label{information_bottleneck_objective}
\min_{p_{T|X}}I(X;T)-\beta I(Y;T)
\end{equation}
given a fixed joint distribution \(p(x,y)\) and a tradeoff parameter \(\beta>0\). The representation \(T\) is constrained to satisfy the Markov chain \(T\rightarrow X\rightarrow Y\), so that \(T\) contains information about \(Y\) only through \(X\). The IB problem in its preceding formulation was first addressed in~\cite{tishby2000informationbottleneckmethod}.

In general, finding the optimal representation \(p_{T|X}\) is a challenging task. Existing exact treatments of IB are confined to two distinct regimes: discrete alphabets and jointly Gaussian variables; the latter being a notable case in which the problem admits closed-form characterization and where the optimal information-preserving representation has a simple parametric structure. %is the Gaussian Information Bottleneck (GIB), where \(X\) and \(Y\) are jointly Gaussian. \textcolor{red}{mention the scalar case as a separate reference.} 
The Gaussian Information Bottleneck (GIB) has been thoroughly studied in~\cite{ChechikGal2005IbfG}: when \(\mathbf X\) and \(\mathbf Y\) are jointly 0-mean Gaussian random vectors, the optimal representation \(\mathbf T = \mathrm{K}\mathbf X+\mathbf N\) is jointly Gaussian with \(\mathbf X\), 
    where \(\mathrm{K}\) is determined by the eigenstructure of the regression matrix associated with canonical correlation analysis, and \(\mathbf N\) is a standard isotropic Gaussian noise independent of \(\mathbf X\). 
    %Motivated by this result, we adopt the same linear form for the encoders in the new LSIB setting.
    
    In addition, the GIB solution exhibits a notion of critical value \(\beta^c\). When \(\beta<\beta^c\), the optimal encoder is one that maximally compresses \(X\) without preserving any information about it: \(I(X;T)=I(Y;T)=0\). As \(\beta\) exceeds \(\beta^c\), the encoder begins to preserve information about \(X\) and \(Y\), and as \(\beta\to\infty\), the encoder retains all the information about \(X\). %Interestingly, a similar notion of critical values also emerges in the LSIB setting.
    
    %in which the optimal representation \({\bf T}\) is jointly Gaussian with \({\bf X}\), and can be represented as a noisy linear projection of \({\bf X}\). The GIB has been thoroughly studied in~\cite{ChechikGal2005IbfG} where the authors provided an important example in which the optimal-information preserving representation has a simple parametric structure.

However, the jointly Gaussian setting does not capture random  quantities whose statistics are better described by heavy-tailed distributions found in many applications~\cite{Tsih1995,crovella1998heavy,clauset2009power}. A canonical class of distributions for modeling such phenomena is the family of symmetric \(\alpha\)-stable (\(\SaS\)) distributions\footnote{A random variable (RV) $X$ is {\em stable\/}  if for any $k \geq 2$ there is $C_k > 0$ and $D_k \in \mathbb{R}$ such that
 $\sum_{j=1}^k X_j \overset{d}{=} C_k X + D_k,$
where $\{X_j\}_j$ are independent copies of $X$. A stable distribution is either $\alpha$-stable or Gaussian.}. These distributions have heavy-tailed power-law densities  and extend the Gaussian law  since they arise as limit distributions of the sum of IID power-law RVs. In contrast, \(\SaS\) variables have infinite variance as opposed to Gaussian RVs that have finite moments of all orders. Examples of such models arise in EEG signals~\cite{11503735}, SAR images~\cite{kuzer2004}, and in the weight distribution of neural networks~\cite{bellido1993,jipeng2024}. 

This motivates studying the IB problem for $\alpha$-stable sources beyond the Gaussian setting.  Motivated by the reduced complexity of linear encoders compared to nonlinear ones and inspired by the GIB framework, we adopt a linear structure for our encoders with additive noise \(N\) belonging to the same statistical family as \(X\). We ask the following: What is the best stable linear encoder and the corresponding Linear Stable Information Bottleneck (LSIB)? and are stable linear encoders optimal for the $\alpha$-stable IB problem? 

In this work, we consider two jointly $\alpha$-stable variables $X$ and $Y$ and investigate the extent to which the structural properties of the GIB solution persist for the $\alpha$-stable case. We provide a closed-form expression for the LSIB and characterize the resulting information tradeoff. Interestingly, a similar notion of critical values $\beta^c$ also emerges in the LSIB setting. Our solution  coincides with the GIB when evaluated at $\alpha = 2$ corresponding to the Gaussian setup. Moreover, we prove that, except for the GIB, stable linear  encoders are strictly suboptimal. Nevertheless, they tend to be asymptotically optimal in the large compression rate regime. Finally, we highlight that with heavier tails (\(\alpha\) decreases),  it becomes more difficult to capture relevant information.

%From classical estimation theory, it is well known that for jointly Gaussian variables, the MMSE estimator has a linear form, i.e., \(\mathbb E[Y\,|\,X]\) is linear in \(X\). Outside the Gaussian setting, however, linear estimators are generally not optimal (when the score function is not linear). In an analogous sense, we show that for \(\alpha<2\), linear stochastic encoders are strictly suboptimal for the IB problem. 

\section{Preliminaries}
\subsection{Information Bottleneck}
    %Given \(\mathbf X\in\mathbb R^{n}\) and \(\mathbf Y\in\mathbb R^{n}\) jointly random vectors, the classical IB formulation aims to minimize the objective function
    %\begin{equation}
     %   \mathcal{L}(\mathbf T)=I(\mathbf X;\mathbf T)-\beta I(\mathbf Y;\mathbf T).
    %\end{equation}
    %over conditional probability laws \(p_{\mathbf T| \mathbf X}\), with \(\mathbf T \in\mathbb R^{n_t}\). Since \(\mathbf T\) is conditionally independent of \(\mathbf Y\) given \(\mathbf X\), \(\mathbf T\rightarrow\mathbf X\rightarrow \mathbf Y\) forms a Markov chain. 
    The standard formulation of the IB problem in the scalar case is given in equation~\eqref{information_bottleneck_objective}. 
    It is straightforward that more compression results in less relevant information preserved and that there is a tradeoff between the two. Therefore, the positive parameter \(\beta\) is called the tradeoff parameter; higher \(\beta\) captures more importance given to information preservation than to data compression. The case \(\beta\in[0,1]\) leads to a degenerate solution: due to the data processing inequality \(I(X;T)-\beta I(Y;T)\ge(1-\beta)I(X;T)\ge0\), and the optimal encoder \(T\) is degenerate: it compresses \(X\) maximally without preserving any information about \(X\) (or \(Y\)), achieving \(I(X;T)=I(Y;T)=0\). In the limit as \(\beta\rightarrow\infty\), no compression is happening and an optimal encoder is $T = X$. 

    Finally, it is worth noting that an alternative but equivalent formulation to~\eqref{information_bottleneck_objective} is the following constrained optimization problem
    \begin{equation}
    \max_{p_{T|X}:I(X;T)\le R} I(Y;T)
    \label{eq:consoptim}
    \end{equation}
    for some fixed rate \(R>0\).
\subsection{Symmetric \(\alpha\)-Stable Laws}
The symmetric \(\alpha\)-stable (\(\SaS\)) distributions represent a natural extension of the Gaussian family to the heavy-tailed regime. A scalar RV \(X\sim\SaS(\gamma)\)  with stability index \(0 < \alpha < 2\) and scale \(\gamma>0\), if its characteristic function 
\begin{equation*}
    \phi_X(\theta)
    =
    \mathbb E\!\left[e^{i\theta X}\right]
    =
    \exp\left\{-\gamma^\alpha |\theta|^\alpha\right\}.
\end{equation*}
In contrast to Gaussian RVs (\(\alpha=2\)), the $\SaS$ variables (\(\alpha<2\)) have power-law probability density functions \(p_X(x) \sim c|x|^{-1-\alpha}\) as \(|x|\to\infty\) and the variables have infinite variance. Thus, the \(\SaS\) model provides a natural setting for studying information-theoretic representations beyond the finite-variance Gaussian regime.

Akin to Gaussian variables, a defining property of \(\SaS\) distributions is the stability property: if \(X_1,\ldots,X_m\) are independent \(\SaS\) RVs with stability index \(\alpha\) and each with scale \(\gamma_i\), then
\begin{equation}
\label{eq:stabilityproperty}
    \sum_{i=1}^m c_iX_i
    \sim
    \SaS\left(
    \gamma
    \right),
    \quad
    \gamma^\alpha=\sum_{i=1}^m |c_i|^\alpha \gamma_i^\alpha,
\end{equation}
for some $\{c_i\} \in \mathbb{R}$. 
%Similarly, scaling preserves the family: if \(X\sim\SaS(\gamma)\), then \(cX\sim\SaS(|c|\gamma)\), for \(c\neq0\). Consequently, if \(Z\sim\SaS(1)\), then
%\begin{equation*}
 %   X\stackrel{d}{=}\gamma Z\implies h(X)=h(Z)+\log\gamma.
%\end{equation*}

Additional properties of \(\SaS\) distributions are provided in Appendix~\ref{app:alpha-stable_probability_laws}. More details on these distributions can be found in~\cite{NolanJohnP2020USDM,samorodnitsky_stable_2017,zolotarev_one-dimensional_1986}.

\section{Main Results}
We consider jointly \(\SaS\) RVs as heavy-tailed analogues of the jointly Gaussian variables underlying the classical GIB setting. The stability index \(\alpha\) therefore provides a parameter controlling the departure from the Gaussian regime, with \(\alpha=2\) recovering the finite-variance Gaussian case.

\subsection{Additive Stable Noise Model}\label{subsec:scalar_channel_model}
We observe a noisy measurement \(X\) of the relevant variable \(Y\), where the measurement is corrupted by additive stable noise \(A\) that is independent of \(Y\). The  model is hence
\begin{equation}
\label{eq:channelmodel}
    X=Y+A,
\end{equation}
where \(Y\sim\SaS(\gamma_Y)\) and \(A\sim\SaS(\gamma_A)\) additive noise independent of \(Y\). By the stability property~\eqref{eq:stabilityproperty}, \(X\) is also \(\SaS\), with scale parameter \(\gamma_X=(\gamma_Y^\alpha+\gamma_A^\alpha)^{1/\alpha}\).  Consequently, the objective is to construct a representation \(T\) of the noisy observation \(X\) that preserves information about the underlying variable \(Y\) while discarding information introduced by the external noise \(A\). 

%In the Gaussian case \(\alpha=2\), we can represent the relation of any joint Gaussian random variables \((X,Y)\) with correlation \(\rho\) as an additive noise channel of the form \(X=aY+N\) where \(N\) is a Gaussian noise independent of \(Y\), with constant \(a\) and variance of \(N\) appropriately chosen. This representation will be useful later when introducing analogues of the signal-to-noise ratio and correlation for the \(\alpha<2\) stable setting.
    
\subsection{Linear Stable Information Bottleneck}
For continuous RVs, a non-degenerate continuous deterministic
representation \(T=f(X)\) generally leads to \(I(X;T)=\infty\), reflecting the
infinite rate required to represent the resulting continuous variable without
randomization. Thus, such deterministic encoders are not suitable for the
finite-rate compression objective of the information bottleneck. Motivated by
the GIB solution, which takes the form of a
noisy linear projection of \(X\), we preserve the linear structure of the encoder while choosing the additive noise \(N\) from the same statistical
family as \(X\) and \(Y\). Specifically, we consider
\begin{equation*}
    T=kX+N,
\end{equation*}
where \(k\) is a scalar and \(N\) is \(\SaS\) noise, independent of \(X\), with a
scale parameter \(\gamma_N\). This choice provides a tractable benchmark for
studying the performance and limitations of linear stochastic encoders under
stable distributions. We refer to encoders of this form as \emph{stable linear encoders}. The optimization of the functional is therefore restricted to a smaller class of conditional laws \(p_{T|X}\):
\begin{equation}\label{eq:LSIB}
    \argmin_{p_{T|X}: T\text{ is a stable linear encoder}} I(X;T)-\beta I(Y;T)
\end{equation}
    
We refer to the optimization problem in~\eqref{eq:LSIB} as the \emph{linear stable information bottleneck} (LSIB), which we study throughout. Indeed, in the Gaussian (\(\alpha=2\)) case, the optimal stable linear encoder for LSIB coincides with the optimal encoder for the original IB problem, since the GIB admits an optimal representation which is linear. For \(\alpha<2\), however, stable linear encoders are not necessarily optimal for the original IB problem. We discuss their optimality in Section~\ref{sec:optimality_LSIB} and provide a negative answer. 
    
Before deriving the LSIB solution, we introduce the following quantities:
\begin{equation}
    \snr:= \frac{|k|^\alpha\gamma_X^\alpha}{\gamma_N^\alpha}\ge0, \quad 
    \rho:=\frac{\gamma_Y}{\gamma_X}\in(0,1). \label{eq:alphath_snr}%\tfrac didn't fix
\end{equation}

In the Gaussian case, these quantities reduce to the conventional signal-to-noise ratio between \(kX\) and \(T\) and the correlation coefficient between \(X\) and \(Y\), respectively. This motivates our terminology: we refer to \(\snr\) as the \emph{\(\alpha^{\rm th}\) signal-to-noise ratio} (\(\alpha^{\rm th}\)-snr) of the linear encoder and to \(\rho\) as the \emph{correspondence} between \(X\) and \(Y\).

\begin{theorem}[Solution to scalar LSIB]\label{thm:optimal_LSIB_scalar}
    A linear encoder \(T=kX+N\) with $N \sim \SaS(1)$ independent of $X$ and 
    \begin{equation}\label{eq:optimal_LSIB_scalar}
        |k|=\begin{cases}
            0, & \beta<\beta^c,\\\displaystyle \frac1{\gamma_X}\left(\frac{\beta\rho^\alpha-1}{1-\rho^\alpha}\right)^{1/\alpha}, & \beta\ge\beta^c,
        \end{cases}\quad \beta^c:=\rho^{-\alpha} > 1,
    \end{equation}
     is an optimal encoder for the LSIB.
\end{theorem}

%\textcolor{green}{Note that \(\rho=1\) implies that \(A=0\) is deterministic and \(\beta^c=1\). In this case, \(I(X;T)-\beta I(Y;T)=(1-\beta)I(X;T)\) and the best we can do, when \(\beta>\beta^c=1\), is to maximize \(I(X;T)\) which is possible by using a deterministic representation \(T=X\) with no compression, or by using linear encoders, taking the limit \(k\to\infty\).}
The optimal stable linear encoder presented above coincides indeed with the GIB solution when \(\alpha=2\). Moreover, similarly to the GIB, the notion of critical values \(\beta^c\) arises naturally in the LSIB setting as well. In particular, as \(\beta\) crosses the critical value \(\beta^c=\rho^{-\alpha}\), the optimal encoder transitions from the degenerate solution \(k=0\) to a nontrivial encoder with \(k\neq0\), corresponding to the activation of a representation component. We now present the proof of Theorem~\ref{thm:optimal_LSIB_scalar}.
 
 \begin{proof}
    Consider $T = k X + N$ with \(N \sim \SaS(\gamma_N)\), and define $Z \sim \SaS(1)$. Then, the mutual information quantities are given by
    \begin{align}
         &I(X;T)\nonumber\\
         =\,&h(T)-h(T\,|\,X) \nonumber\\
         =\,&h(kX+N)-h(N)\nonumber\\
=\,&h(Z)+\frac1\alpha\log\left(|k|^\alpha\gamma_X^\alpha+\gamma_N^\alpha\right)-h(Z)-\log(\gamma_N) \label{eq:stabpropuse}\\
=\,&\frac1\alpha\log(1+\snr),\label{eq:compression_rate_scalar}
\end{align}
and
\begin{align}
        &I(Y;T)\nonumber\\
        %=\,&h(T)-h(T\,|\,Y)\nonumber\\
        =\,&h(kX+N)-h(kA+N)\nonumber\\
        =\,&h(Z)+\frac{1}{\alpha}\log\left(|k|^\alpha\gamma_X^\alpha+\gamma_N^\alpha\right)-h(Z)\nonumber\\
        &\qquad \qquad \qquad \qquad -\frac{1}{\alpha}\log(|k|^\alpha\gamma_A^\alpha+\gamma_N^\alpha)\label{eq:stabpropuse1}\\
=\,&\frac1\alpha\log\left(\frac{1+|k|^\alpha\gamma_X^\alpha\gamma_N^{-\alpha}}{1+|k|^\alpha\gamma_X^\alpha\gamma_N^{-\alpha}(1-\gamma_Y^\alpha\gamma_X^{-\alpha})}\right)\nonumber\\
        =\,&\frac1\alpha\log\left(\frac{1+\snr}{1+\snr(1-\rho^\alpha)}\right),\label{eq:relevant_information_scalar}
    \end{align}
    where we used equation~\eqref{eq:stabilityproperty} to write equations~\eqref{eq:stabpropuse} and~\eqref{eq:stabpropuse1}. Since equations~\eqref{eq:compression_rate_scalar} and~\eqref{eq:relevant_information_scalar} are given in terms of \(\snr\), we minimize \(\mathcal{L}(\snr):=\mathcal L(T)=I(X;T)-\beta I(Y;T)\) over \(\snr\ge0\) %Since that \(\mathcal{L}(\snr)\in C^\infty(\mathbb R)\) \textcolor{red}{define this as a footnote}. When the tradeoff parameter \(\beta\) exceeds \(\beta^c=\rho^{-\alpha}\), solving for \(\frac{d\mathcal L}{d\snr}=0\) yields a unique solution \(\snr>0\) which achieves the global minimum. Otherwise, when \(\beta<\beta^c\), \(\snr=0\) is the unique global minimizer. Finally, we choose \(\gamma_N=1\) and the appropriate \(k\) to achieve the optimal \(\snr\). The detailed steps of the proof are provided in Appendix~\ref{app:proofs}.    
    by taking the first derivative 
and setting it to zero 
\begin{align}
\frac{1-\beta}{1+\snr}
+\frac{\beta(1-\rho^\alpha)}
{1+\snr(1-\rho^\alpha)}
&=0, \nonumber\\
\frac{1+\snr}
{1+\snr(1-\rho^\alpha)}
&=\frac{\beta-1}
{\beta(1-\rho^\alpha)}.
\label{eq:stationary}
\end{align}
Solving for \(\snr\) yields
\begin{equation}
\snr^*
=\frac{\beta\rho^\alpha-1}
{1-\rho^\alpha},
\label{eq:stationary_snr}
\end{equation}
where \(\snr^*\ge0\) denotes the point satisfying the
stationary equation whenever such a point exists.

To determine when a stationary point exists, define the left-hand
side of~\eqref{eq:stationary} as \(R(\snr)\),
%\[
%R(\snr):=
%\frac{1+\snr}
%{1+\snr(1-\rho^\alpha)}.
%\]
 which is strictly increasing over
\(\snr\ge0\), with range
\(\mathcal D=[1,(1-\rho^\alpha)^{-1})\). Hence, whenever a
solution to~\eqref{eq:stationary} exists, it is unique.

When \(\beta< \beta^c \, (=\rho^{-\alpha})\), the right-hand side of
\eqref{eq:stationary} satisfies
\(
\tfrac{\beta-1}{\beta(1-\rho^\alpha)}
\notin\mathcal D.
\)
Thus, no stationary point exists. In this case,
\(d\mathcal L/d\snr>0\) for all \(\snr\ge0\), and therefore
\(\snr=0\) is the unique global minimizer.

When \(\beta\ge\beta^c\), the value \(\snr^*\) in
\eqref{eq:stationary_snr} is a stationary point. To determine
its optimality, rewrite the derivative as
\begin{equation*}
\frac{d\mathcal L}{d\snr}
= \frac{1}{\alpha}
\frac{1-\beta\rho^\alpha
+\snr(1-\rho^\alpha)}
{(1+\snr)\bigl[1+\snr(1-\rho^\alpha)\bigr]}.
\end{equation*}

The denominator is positive for \(\snr\ge0\), while the numerator
is strictly increasing in \(\snr\) and vanishes uniquely at
\(\snr=\snr^*\). Hence, for \(\beta>\beta^c\), \(\snr^*\) is the unique global minimizer.
%\[
%\frac{d\mathcal L}{d\snr}<0
%\quad\text{for }0\le\snr<\snr^*,
%\quad
%\frac{d\mathcal L}{d\snr}>0
%\quad\text{for }\snr>\snr^*.
%\]
%Therefore, \(\snr^*\) is the unique global minimizer. 
At
\(\beta=\beta^c\), we have \(\snr^*=0\), and the derivative is
strictly positive for every \(\snr>0\), so \(\snr=0\) is again
the unique global minimizer.

Finally, one can achieve the optimal \(\snr\) by choosing \(\gamma_N=1\) and the appropriate value of \(k\). 
    
We note that \(\snr^*\) is a strictly increasing function of \(\beta\) when \(\beta\ge\beta^c\), which is expected: the more weight we assign to the relevant information preservation as opposed to information compression, the higher the required value of \(\snr\).
\end{proof}

\begin{remark}
 %When \(\alpha=2\), the optimal linear encoder as given in Theorem~\ref{thm:optimal_LSIB_scalar} indeed collapses to the optimal encoder of GIB given in~\cite{ChechikGal2005IbfG}. 
 In comparison to the GIB, one can identify  the corresponding analogous quantities for ``signal-to-noise ratio" and ``correlation", which we called ``\(\alpha^{\rm th}\)-snr" and the ``correspondence", respectively, that naturally arise from the expressions for mutual information \(I(X;T)\) and \(I(Y;T)\).
\end{remark}

\subsection{LSIB Information Curve}
The solution to the LSIB problem allows us to characterize the tradeoff between compression and relevant information. More precisely, we can determine the maximal relevant information \(I(Y;T)\) achievable for a given compression rate \(I(X;T)\) when using stable linear encoders. This tradeoff is described by the \emph{LSIB information curve}. 

When the tradeoff parameter is less than the critical value \(\beta<\beta^c=\rho^{-\alpha}\), the optimal encoder is degenerate with \(k=0\), and therefore this corresponds to the origin of the LSIB information curve where \(I(X;T)=I(Y;T)=0\). When \(\beta>\beta^c\), substituting the optimal \(\alpha^{\rm th}\)-snr from \eqref{eq:alphath_snr}, \eqref{eq:optimal_LSIB_scalar} into the mutual information expressions in \eqref{eq:compression_rate_scalar}, \eqref{eq:relevant_information_scalar} yields
\begin{align}
    I_\beta(X;T)
    &=\frac{1}{\alpha}
    \log\left(
        \frac{\rho^\alpha(\beta-1)}
        {1-\rho^\alpha}
\right),\label{eq:LSIB_information_curve_compression}\\
    I_\beta(Y;T)
    &=\frac{1}{\alpha}
    \log\left(
        \frac{\beta-1}
        {\beta(1-\rho^\alpha)}
    \right).\label{eq:LSIB_information_curve_relevant_information}
\end{align}

Hence, the LSIB information curve is fully characterized by the stable
parameter \(\alpha\) and the correspondence \(\rho\),
and can be traced parametrically by varying the tradeoff parameter \(\beta\) over \([\beta^c,\infty)\). As depicted in Figure~\ref{fig:LSIB}, the LSIB information curve is indeed increasing and concave, reflecting the decreasing returns in the relevant information obtained by increasing the amount of information retained about \(X\).

We note that as \(\beta\) increases, the encoder places progressively greater emphasis on preserving relevant information rather than compression, resulting in increasing values of both \(I(X;T)\) and \(I(Y;T)\). At the critical value \(\beta=\beta^c\), both mutual information are zero. Thus, the nontrivial portion of the information curve emerges continuously from the origin when \(\beta\) crosses the critical value. A similar behavior is observed in the GIB~\cite{ChechikGal2005IbfG}. 

In the limiting case, consider \(\beta\to\infty\), %In this limit, we aim to preserve the input data as much as possible without compressing it, 
for which \(
    I(X;T)\rightarrow\infty\) and
    \(I(Y;T)\rightarrow\frac{1}{\alpha}\log(\frac{1}{1-\rho^\alpha}) = I_{\mathrm{sat}}(\alpha,\rho).\)
The limiting value of the relevant information preserved \(I(Y;T)\) is precisely the mutual information between \(X\) and \(Y\): it is not unexpected that \(\lim_{\beta\to\infty}I_\beta(Y;T)=I(X;Y)\) 
shows that with infinite-rate, stable linear encoders can preserve all the information about \(Y\) contained in \(X\).

It is also useful to express the relevant information preserved \(I(Y;T)\) directly as a
function of the compression rate \(R:=I(X;T)\). Eliminating
\(\beta\) from~\eqref{eq:LSIB_information_curve_compression} and \eqref{eq:LSIB_information_curve_relevant_information} yields
\begin{equation}\label{eq:LSIB_information_curve_function_of_R}
    I(Y;T)
    =
    \frac{1}{\alpha}\log\left(\frac{e^{\alpha R}}{\rho^\alpha+(1-\rho^\alpha)e^{\alpha R}}\right).
\end{equation}
Thus, the LSIB information curve admits both a parametric
representation in terms of \(\beta\) and an explicit representation as
a function of the compression rate \(R\).

\begin{figure*}[!t]
    \centering

    \subfloat[]{
        \includegraphics[width=0.3\textwidth]{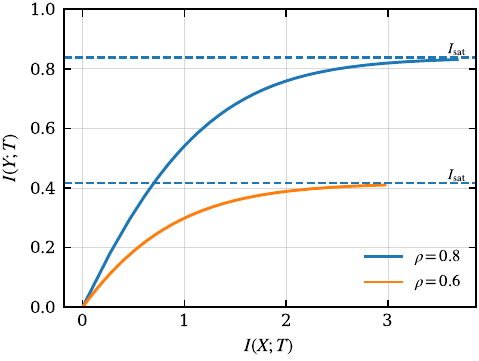}
        \label{fig:LSIB}
    }
    \hfill
    \subfloat[]{
        \includegraphics[width=0.3\textwidth]{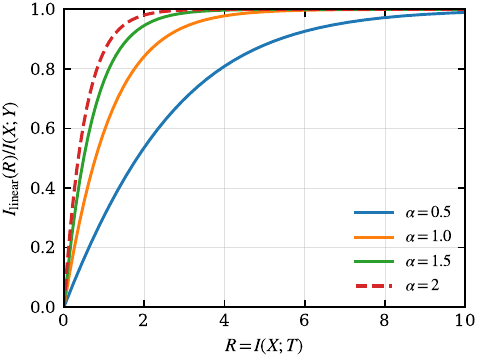}
\label{fig:fraction_relevant_information_preserved_linear_encoder_1}
    }
    \hfill
    \subfloat[]{
        \includegraphics[width=0.3\textwidth]{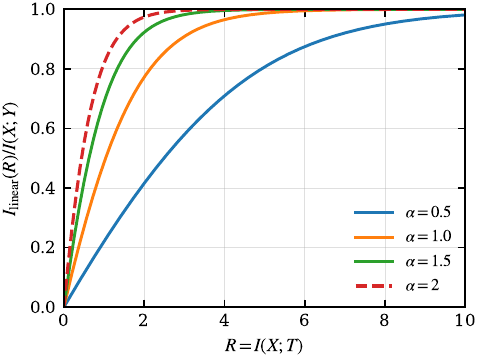}
\label{fig:fraction_relevant_information_preserved_linear_encoder_2}
    } 
    \caption{(a) Information curve plotting the mutual information as the tradeoff parameter \(\beta\) varies for \(\alpha=1.5, \rho=(0.6,0.8)\). (b) Fraction of relevant information preserved for \(\alpha=(2.0,1.5,1.0,0.5)\), \(\rho=0.3\). (c) \(\rho=0.7\).}
    \label{fig:three_cases}
\end{figure*}

%\begin{figure}[ht]
 %   \centering
%\includegraphics[width=0.7\columnwidth]{Figures/information_curve_rho_comparison.pdf}
 %   \caption{Information curve plotting the mutual information as the tradeoff parameter \(\beta\) varies for \(\alpha=1.5, \rho=(0.6,0.8)\).}
  %  \label{fig:LSIB}
%\end{figure}

\section{Non-Optimality of stable Linear Encoders}\label{sec:optimality_LSIB}
The LSIB only considers the class of stable linear encoders which are known to be optimal for GIB (\(\alpha=2\)). The optimality in the Gaussian case mainly relies on the tightness of the entropy power inequality (EPI) under Gaussian variables, and thus the proof technique is rather insightful only for the GIB (see, e.g.,~\cite{dikshtein2022informationbottleneckgaussianprocesses}). Thus, for $\alpha$-stable variables ($\alpha < 2$), the optimal stable linear encoder may be suboptimal. In the remainder of this section, we positively answer this question and quantify the performance of the LSIB and its gap from optimality. 

We note that this result, though seemingly natural, should not be taken for granted. As an example, it has been recently shown that the conditional mean estimators for the relevant data \(Y\) based on observed data \(X\) when both are jointly \(\SaS\), exist and are linear~\cite{Chouity2026Linear,11653965} similar to Gaussian case. It may be rather unexpected that stable linear encoders optimal for LSIB may be strictly suboptimal for the original IB problem.

\begin{theorem}
\label{th:suboptimal}
    Consider $(X,Y)$ two jointly $\SaS$ RVs as in~\eqref{eq:channelmodel}. When \(\beta>\beta^c=\rho^{-\alpha}\), any stable linear encoder \(T=kX+N\) with \(k\neq0\) and \(N\sim\SaS(\gamma_N)\) independent of \(X\) is strictly suboptimal for the general IB problem, i.e. it cannot achieve \(\min \,[I(X;T)-\beta I(Y;T)]\).
    %\(\min_{p_{T|X}}I(X;T)-\beta I(Y;T)\).
\end{theorem}

\begin{proof}
    Since $X$ is symmetric, we assume, Without Loss of Generality (WLoG), \(k>0\). An optimal encoder with conditional law \(p_{T|X}\) for the general IB problem satisfies the stationary condition given in~\cite{tishby2000informationbottleneckmethod}:
\begin{equation}\label{eq:selfconsistent}
      p_{T|X}(t| x)=\frac{p_T(t)}{Z(x,\beta)}\exp\!\left\{\!-\beta
      D_{\mathrm{KL}}(p_{Y|X}(y|x)\,\|\,p_{Y|T}(y| t))\right\},
    \end{equation}
    where $Z(x,\beta)$ is the normalization term. Assume, for the sake of contradiction, that the linear encoder \(T=kX+N\) is optimal, and satisfies~\eqref{eq:selfconsistent}. Define for any \(x_1\neq x_2\), the function
    \begin{align} \xi(t,x_1,x_2)&:=\log p_N(t-kx_1)-\log p_N(t-kx_2)-\beta\int_\mathbb R\Delta(y,x_1,x_2) \log p_{Y|T}(y|t)\,dy, \label{eq:Xi}
    \end{align}
    where
    \begin{eqnarray}
\Delta(y,x_1,x_2)\!\!\!&:=&\!\!\!p_{Y|X}(y|x_1)-p_{Y|X}(y|x_2) \nonumber\\
\!\!\!&=&\!\!\!p_Y(y)\left[\frac{p_A(x_1-y)}{p_X(x_1)}-\frac{p_A(x_2-y)}{p_X(x_2)}\right] \label{eq:deltay}
    \end{eqnarray}
    
We show in Lemma~\ref{lem:stationarity} in Appendix~\ref{app:proofs}, that if the stable linear encoder satisfies~\eqref{eq:selfconsistent}, then 
the function $\xi_{\mathrm{lin,opt}}(t,x_1,x_2)$ must be constant in $t$, and hence \(\partial_t\xi_{\mathrm{lin,opt}}(t,x_1,x_2)=0\) for all \(t\). However, this is contradicted by the following asymptotic result for the stable linear encoder given in Lemma~\ref{lem:stationarity-asymptotics} in Appendix~\ref{app:proofs}
    \begin{equation}
    \lim_{t\to\infty}t^{2}\partial _t\xi_{\mathrm{lin}}(t,x_1,x_2)
      %&=-(1+\alpha)k(x_1-x_2)+\beta(1+\alpha)k\rho^\alpha(x_1-x_2)\nonumber\\
      =(1+\alpha)k(x_1-x_2)\bigl[\beta\rho^\alpha-1\bigr]\nonumber
    \end{equation}
    which is nonzero since \(\alpha, k >0\), \(x_1\neq x_2\), and \(\beta>\beta^c=\rho^{-\alpha}\). Therefore, the stable linear encoder fails the necessary condition for optimality of encoders \(p_{T|X}(t|x)\) and is therefore strictly suboptimal.\footnote{Note that, for the Gaussian case, \(\partial_t\xi_{\mathrm{lin}}(t,x_1,x_2)=k(x_1-x_2)[1/\sigma_N^2-\beta\rho^2/\sigma_M^2]\) where \(\sigma_M^2:=k^2\sigma_N^2+\sigma_A^2\). Using the optimal stable linear encoder of Theorem~\ref{thm:optimal_LSIB_scalar}, \(\sigma_M^2=2\beta\rho^2\) and \(\sigma_N^2=2\gamma_N^2=2\), so \(\partial_t\xi_{\rm lin}(t,x_1,x_2)=0\) for all \(t\) and thus satisfies the necessary condition of optimality.}
\end{proof}

Since stable linear encoders are provably not optimal when $\alpha < 2$, we wish to assess their performance loss. To that effect, we consider the constrained optimization formulation given by equation~\eqref{eq:consoptim}. For a constrained compression rate \(I(X;T)\le R\), we attempt to quantify sub-optimality in terms of the information preservation of the relevance variable \(Y\). Define
\begin{equation*}
\begin{aligned}
\mathrm{OPT}(R)
=\sup\bigl\{ I(Y;T):\;& I(X;T)\le R,\\
& T \text{ is a stochastic encoder of } X \bigr\}.
\end{aligned}
\end{equation*}
 Recall that with an optimal stable linear encoder \(T_L\), the maximum \(I(Y;T_L)\) denoted \(I_{\rm linear}(R)\) attained for a fixed compression rate \(I(X;T_L)=R\) is given in equation~\eqref{eq:LSIB_information_curve_function_of_R}, which is rewritten as follows:
\begin{equation*}
    I_{\rm linear}(R)=\frac{1}{\alpha}\log\left(\frac{1+s}{1+s(1-\rho^\alpha)}\right)
\end{equation*}
where \(s:=e^{\alpha R}-1\).
Define the \emph{relevant information loss} that results from using the stable linear encoder as
\begin{equation*}
    \ell(R):=\mathrm{OPT}(R)-I_{\rm linear}(R)\ge0,
\end{equation*}
to which we provide an upper bound in the following corollary.
\begin{corollary}\label{thm:gap_bound}
    At a fixed compression rate \(I(X;T)=R\), the relevant information loss $\ell(R)$ due to the optimal stable linear encoder is upperbounded by
    \begin{equation*}
        \ell(R) \leq \frac{1}{\alpha}\log\left[\frac{1+s(1-\rho^\alpha)}{(1-\rho^\alpha)(1+s)}\right]
    \end{equation*}
    where \(s:=e^{\alpha R}-1\).
\end{corollary}
\begin{proof}
    Since \(Y\rightarrow X\rightarrow T\) is a Markov chain, then by the data processing inequality \(I(Y;T)\le I(Y;X)=\frac{1}{\alpha}\log\frac{\gamma_X^\alpha}{\gamma_A^\alpha}=-\frac{1}{\alpha}\log(1-\rho^\alpha)\), and the bound follows. 
\end{proof}

%We can normalize the relevant information lost by \(I(X;Y)\) and view the relative relevant information loss instead:
%\begin{align}
 %   \frac{\mathrm{gap}(R)}{I(X;Y)}&\le1-\frac{I_{\rm linear}(R)}{I(X;Y)}\\
  %  &=-\frac{\log\left(1+\frac{\rho^\alpha}{1-\rho^\alpha}e^{-\alpha R}\right)}{\log(1-\rho^\alpha)}
%\end{align}
We plot in Figures~\ref{fig:fraction_relevant_information_preserved_linear_encoder_1} and~\ref{fig:fraction_relevant_information_preserved_linear_encoder_2}, a lowerbound on the fraction of relevant information preserved at a rate \(R\)
\begin{equation}
\label{eq:gap}
   \frac{I_{\rm linear}(R)}{\mathrm{OPT}(R)} \geq  \frac{I_{\rm linear}(R)}{I(X;Y)}
    =\frac{\log\left(1-\rho^\alpha+\rho^\alpha e^{-\alpha R}\right)}{\log(1-\rho^\alpha)},
\end{equation}
 for different values of \(\alpha\) and \(\rho\). As depicted, the lower bound converges to $1$ at large compression rate $R$. Indeed, as implied by equation~\eqref{eq:gap}, the gap is upperbounded by  
 \[1 -  \frac{I_{\rm linear}(R)}{\mathrm{OPT}(R)} \leq -\frac{\rho^{\alpha}}{(1 - \rho^\alpha)\log(1-\rho^\alpha)} e^{-\alpha R},\]
 which exponentially decays to $0$ with $R$.  Moreover, we observe that at a fixed \(R\) the fraction of relevant information retained decreases as the tails of the $\SaS$ density become heavier (i.e., the stability index \(\alpha\) is smaller). This is intuitive since the heavier the variable, the ``harder" it is to compress it.  
 
 Additionally, Figures~\ref{fig:fraction_relevant_information_preserved_linear_encoder_1} and~\ref{fig:fraction_relevant_information_preserved_linear_encoder_2} show that the fraction of relevant information preserved increases as the correspondence parameter \(\rho\) decreases and the stability index \(\alpha\) approaches \(2\). Consequently, in these regimes, most of the relevant information can be retained at moderate rates \(R\), with only a small corresponding loss \(\ell(R)\).

%\begin{figure}[ht]
 %   \centering
  %  \includegraphics[width=0.5\textwidth]{Figures/information_fraction_rho_03.pdf}
   % \caption{Fraction of relevant information preserved for \(\alpha=(2.0,1.5,1.0,0.5)\), \(\rho=0.3\).}
%\label{fig:fraction_relevant_information_preserved_linear_encoder}
%\end{figure}

%\begin{figure}[ht]
 %   \centering
  %  \includegraphics[width=0.5\textwidth]{Figures/information_fraction_rho_07.pdf}
   % \caption{Fraction of relevant information preserved for \(\alpha=(2.0,1.5,1.0,0.5)\), \(\rho=0.7\).}
%\label{fig:fraction_relevant_information_preserved_linear_encoder}
%\end{figure}

\section{Conclusion}
In this work, we studied the IB problem for jointly $\alpha$-stable variables related through an additive model. We showed that unlike the jointly Gaussian case, stochastic stable linear encoders are not optimal in general. Moreover, we investigated the performance of the best linear encoders of the form $T = k X + N$, where $N$  is a stable RV by evaluating the LSIB in closed-form and identifying its critical value $\beta^c$. The LSIB expression reduces to the optimal solution for the jointly Gaussian problem as a special case. For the $\alpha$-stable setup, it is shown to behave asymptotically optimal at large rate $R$.

\appendices
\section{\(\alpha\)-Stable Probability Laws}\label{app:alpha-stable_probability_laws}
Alpha-stable distributions are limiting distributions of normalized IID sums under the Generalized Central Limit Theorem~\cite{kolmo1968}. They are heavy-tailed, power-law variables with infinite variance, and even some have non-finite expectations. Next, we present some useful properties.
%\begin{definition}
 %   A random variable \(X\) is said to be \(\alpha\)-\emph{stable} if for \(X_1\) and \(X_2\) independent copies of \(X\) and any positive constants \(a\) and \(b\),
  %  \begin{equation}
   %     aX_1+bX_2\overset d= cX+d
   % \end{equation}
    %for some positive \(c\) and some \(d\in\mathbb{R}\).
%\end{definition}
\begin{definition}~\cite{NolanJohnP2020USDM}.
    A random variable \(X\) is $\alpha$-stable iff \(X \overset d= aZ+b\), where \(0<\alpha<2,-1\le\beta\le1,a\neq0,b\in\mathbb{R}\) and \(Z\) is a random variable with the characteristic function
    \begin{align*}
\label{def:standardized_stable}
    \phi(\theta)&=\mathbb{E}[e^{i\theta Z}]\nonumber\\
    &=\begin{cases}
            \exp \left\{-|\theta|^\alpha\left[1-i\beta\tan\frac{\pi\alpha}{2}(\mathrm{sgn}\,\theta)\right]\right\}, & \alpha\neq1, \\
            \exp \left\{-|\theta|\left[1+i\beta\frac{2}{\pi}(\mathrm{sgn}\,\theta)\log|\theta|\right]\right\}, & \alpha=1.
        \end{cases}
    \end{align*}

   Using the standard parametrization, we denote \(X\) is \(S(\alpha,\beta,\gamma,\delta)\) if
    \begin{equation*}
        X\overset d= \begin{cases}
            \gamma Z +\delta & \alpha\neq1\\
            \gamma Z + (\delta + \beta\frac2{\pi}\gamma \log \gamma) & \alpha=1.
        \end{cases}
    \end{equation*}
     Hence, \(X\) has the characteristic function
    \begin{align*}
&\phi(\theta)\\
&=\begin{cases}
            \exp \left\{-\gamma^\alpha|\theta|^\alpha\left[1-i\beta\tan\frac{\pi\alpha}{2}(\mathrm{sgn}\,\theta)\right]+i\delta \theta\right\}, & \alpha\neq1, \\
            \exp \left\{-\gamma|\theta|\left[1+i\beta\frac{2}{\pi}(\mathrm{sgn}\,\theta)\log|\theta|\right]+i\delta\theta\right\}, & \alpha=1.
            \end{cases}
    \end{align*}
\end{definition}
%The usual restrictions on the parameters are imposed: \(\alpha\in(0,2],\beta\in[-1,1],\gamma\in(0,+\infty),\text{ and }\delta\in\mathbb{R}\). If \(\gamma=0\), we consider it to be a degenerate random variable which is deterministic. Note that the parameter \(k\) describes how the family of \(\alpha\)-stable laws are parametrized (see \cite{NolanJohnP2020USDM}), and we will only rely on the parametrization \(k=1\). So, we can drop the parametrization \(k\) and write \(\mathbf S(\alpha,\beta,\gamma,\delta)\) for brevity.

If \(\beta=0\) and \(\delta=0\), we say that \(X\) is symmetric \(\alpha\)-stable and we denote \(X\sim\SaS(\gamma)\). Its characteristic function is equal to
\begin{equation}
\phi(\theta) = e^{-\gamma^\alpha|\theta|^\alpha}, \quad 0 < \alpha < 2.
\label{eq:charsym}
\end{equation}

\begin{property}[Series Expansion]\label{prop:series_expansion}
The density of \(X\sim\SaS(\gamma)\) admits the following asymptotic series expansion (\(x\to+\infty\)) (see \cite{NolanJohnP2020USDM},\cite{zolotarev_one-dimensional_1986}):
\begin{equation}
\label{eq:expansion}
p_X(x)\sim\frac1\pi\sum_{n=1}^{\infty}\frac{(-1)^{n+1}}{n!}\Gamma(n\alpha+1)\sin\bigl(\tfrac{n\pi\alpha}{2}\bigr)\gamma^{n\alpha}x^{-n\alpha-1}.
\end{equation}
%For the first derivative of the density, the asymptotic series expansion \((x\to\infty)\) is given as follows:
%\begin{equation}
%p_X'(x)\sim\frac1\pi\sum_{n=1}^\infty\frac{(-1)^n}{n!}\Gamma(n\alpha+2)\sin\!\Big(\frac{n\pi\alpha}2\Big)\gamma^{n\alpha}x^{-n\alpha-2}.
%\end{equation}
The asymptotic series expansion \((x\to+\infty)\) of the \(k^{\rm th}\) derivative of the density, $k \geq 1$, is
\begin{equation*}
p_X^{(k)}(x)\sim\sum_{n=1}^{\infty} c_{k,n}x^{-n\alpha-k-1},
\end{equation*}
where
\begin{equation*}
c_{k,n}:=\frac{(-1)^{k}}\pi\frac{(-1)^{n+1}\gamma^{n\alpha}}{n!}\,\Gamma(n\alpha+k+1)\,\sin\!\Big(\frac{n\pi\alpha}{2}\Big).
\end{equation*}
\end{property}
The asymptotic expansion as \(x\to-\infty\) may be easily obtained from the fact that \(p_X^{(k)}\) is even (odd) when \(k\) is even (odd).
\begin{proof}
The expansion in~\eqref{eq:expansion} is obtained from the Fourier representation of the density function given in equation~\eqref{eq:charsym}
\begin{align*}
p_X(x)&=\frac{1}{2\pi}\int_{-\infty}^\infty e^{-\gamma^\alpha|\theta|^\alpha}e^{-i\theta x}\,d\theta\\
&=\frac1{\pi}\int_0^\infty e^{-\gamma^\alpha \theta^\alpha} \cos(\theta x)\,d\theta.
\end{align*}

Differentiating the Fourier representation in~\eqref{eq:charsym}
under the integral sign $k$ times, which is justified by dominated
convergence, yields
\begin{align}
p_X^{(k)}(x)
&=\frac{1}{2\pi}\int_{-\infty}^{\infty}
(-i\theta)^k e^{-\gamma^\alpha|\theta|^\alpha}
e^{-i\theta x}\,d\theta \nonumber\\
&=\frac{1}{\pi}\Re\left\{
i^k\int_0^\infty
r^k e^{-\gamma^\alpha r^\alpha}e^{irx}\,dr
\right\}.
\label{eq:pk-fourier}
\end{align}

Let
\[
f(z):=z^k e^{-\gamma^\alpha z^\alpha}.
\]
Fix
\[
0<\psi_0<
\min\left\{\frac{\pi}{2},\frac{\pi}{2\alpha}\right\},
\]
and take $z^\alpha$ on its principal branch. Then $f$ is analytic in
$0<\arg z<\psi_0$ and continuous up to the boundary away from the
origin. For $z=Re^{i\psi}$, $0\le\psi\le\psi_0$, and $x>0$,
\begin{align}
\left|e^{ixz}\right|
&=e^{-xR\sin\psi}\le 1,
\label{eq:bd1}\\
|f(z)|
&=R^k e^{-\gamma^\alpha R^\alpha\cos(\alpha\psi)}
\nonumber\\
&\le
R^k e^{-\gamma^\alpha R^\alpha\cos(\alpha\psi_0)},
\label{eq:bd2}
\end{align}
where $\cos(\alpha\psi_0)>0$.

To account for the branch point at the origin when $\alpha$ is
noninteger, consider first the truncated sector contour consisting of
the segment $[\varepsilon,R]$, the arc
$\{Re^{i\psi}:0\le\psi\le\psi_0\}$, the ray from
$Re^{i\psi_0}$ to $\varepsilon e^{i\psi_0}$, and the inner arc
$\{\varepsilon e^{i\psi}:0\le\psi\le\psi_0\}$. By Cauchy's integral
theorem,
\[
\oint f(z)e^{ixz}\,dz=0.
\]
Using~\eqref{eq:bd1}--\eqref{eq:bd2}, the contribution of the outer
arc satisfies
\begin{equation*}
\left|
\int_0^{\psi_0}
iRe^{i\psi}f(Re^{i\psi})
e^{ixRe^{i\psi}}\,d\psi
\right|
\le
\psi_0 R^{k+1}
e^{-\gamma^\alpha R^\alpha\cos(\alpha\psi_0)}
\to0
\end{equation*}
as $R\to\infty$. The contribution of the inner arc is
$O(\varepsilon^{k+1})$ and therefore vanishes as
$\varepsilon\to0$. Hence,
\begin{equation}
\int_0^\infty f(r)e^{ixr}\,dr
=
e^{i\psi_0}
\int_0^\infty
f(re^{i\psi_0})
e^{ixre^{i\psi_0}}\,dr.
\label{eq:RHS}
\end{equation}

We next evaluate the right-hand side of~\eqref{eq:RHS} using the
generalized Watson Lemma~\cite{Wyman1964Laplace}, \cite{wong_asymptotic_1989}. Around the origin,
\[
f(r)
=
\sum_{n=0}^\infty
\frac{(-1)^n\gamma^{n\alpha}}{n!}
r^{n\alpha+k}
=
\sum_{n=0}^\infty a_n r^{\lambda_n-1},
\]
where
\[
a_n:=\frac{(-1)^n\gamma^{n\alpha}}{n!},
\qquad
\lambda_n:=n\alpha+k+1.
\]
Along the rotated ray,
\[
e^{i\psi_0}f(re^{i\psi_0})
=
\sum_{n=0}^\infty
a_n e^{i\lambda_n\psi_0}r^{\lambda_n-1}.
\]
Moreover,
\(
e^{ixre^{i\psi_0}}
=
e^{-z r}\), where \(z:=-ixe^{i\psi_0}\),
and
\[
|\arg z|
=
\frac{\pi}{2}-\psi_0<\frac{\pi}{2}.
\]
The conditions of the generalized Watson Lemma therefore hold, and
as $x\to\infty$,
\begin{align}
\int_0^\infty f(r)e^{ixr}\,dr
&\sim
\sum_{n=0}^\infty
a_n e^{i\lambda_n\psi_0}
\Gamma(\lambda_n)z^{-\lambda_n}
\nonumber\\
&=
\sum_{n=0}^\infty
a_n\Gamma(\lambda_n)
x^{-\lambda_n}e^{i\lambda_n\pi/2}.
\label{eq:watson-k}
\end{align}

Substituting~\eqref{eq:watson-k} into~\eqref{eq:pk-fourier} gives
\[
p_X^{(k)}(x)
\sim
\frac{1}{\pi}
\sum_{n=0}^\infty
a_n\Gamma(\lambda_n)x^{-\lambda_n}
\Re\left\{
i^k e^{i\lambda_n\pi/2}
\right\}.
\]
Since
\begin{align*}
i^k e^{i\lambda_n\pi/2}
&=
e^{ik\pi/2}
e^{i(n\alpha+k+1)\pi/2}\\
&=
(-1)^k i e^{in\alpha\pi/2},
\end{align*}
we have
\[
\Re\left\{
i^k e^{i\lambda_n\pi/2}
\right\}
=
(-1)^{k+1}
\sin\left(\frac{n\pi\alpha}{2}\right).
\]
The $n=0$ term therefore vanishes, and
\begin{equation*}
p_X^{(k)}(x) \sim
\!\sum_{n=1}^\infty
\frac{(-1)^{n+k+1}\gamma^{n\alpha}}{\pi n!}
\Gamma(n\alpha+k+1)
\sin\left(\frac{n\pi\alpha}{2}\right)
x^{-n\alpha-k-1},
\end{equation*}
as $x\to+\infty$.
\end{proof}

\begin{proposition}[Tail expansion of the score]\label{prop:score}
    Let \(X\sim\SaS(\gamma)\). Define the \emph{score} function for the density of \(X\) as \(s_X(u):=\frac{d}{du}\log p_X(u)=\frac{p_X'(u)}{p_X(u)}\).
    The function \(\zeta_X(u):=s_X(u)+\frac{1+\alpha}{u}\) satisfies, for \(|u|\ge1\),
    \begin{equation}\label{eq:Phi-bounds}
\begin{aligned}
|\zeta_X(u)|
&\le C_X |u|^{-1-\alpha}, \qquad
|\zeta_X'(u)|
\le C_X |u|^{-2-\alpha},\\
|\zeta_X''(u)|
&\le C_X |u|^{-3-\alpha},
\end{aligned}
\end{equation}
for some \(C_X<\infty\). Equivalently,
\begin{equation}\label{eq:score-exp}
\begin{aligned}
s_X(u)
&=-\frac{1+\alpha}{u}
+O(|u|^{-1-\alpha}),\\
s_X'(u)
&=\frac{1+\alpha}{u^2}
+O(|u|^{-2-\alpha}),\\
s_X''(u)
&=-\frac{2(1+\alpha)}{u^3}
+O(|u|^{-3-\alpha}).
\end{aligned}
\end{equation}
    Moreover, \(\sup_u|s_X(u)|<\infty\) and \(|s_X(u)|\le C/|u|\) for \(|u|\ge1\).
\end{proposition}

%\begin{proof}[Proof sketch]
%The proof follows from the asymptotic series expansion given in Property~\ref{prop:series_expansion}. Positivity, smoothness and boundedness of
%\(s_X(\cdot)\) are standard. 
%\end{proof}

\begin{proof}
    From Property~\ref{prop:series_expansion},
    \begin{equation*}
        p_X(u)=c_{0,1}u^{-\alpha-1}(1+r_0(u))
    \end{equation*}
    where \(r_0(u):=O(u^{-\alpha})\) as \(u\to\infty\), and 
    \begin{equation*}
        p_X^{(k)}(u)=c_{k,1}u^{-\alpha-k-1}(1+r_k(u))
    \end{equation*}
    where \(r_k(u)=O(u^{-\alpha})\) as \(u\to\infty\). Therefore, the score function
    \begin{align*}
        s_X(u)=\frac{p_X'(u)}{p_X(u)}&=\frac{c_{1,1}}{c_{0,1}}\,u^{-1}\cdot\frac{1+r_1}{1+r_0}\\
        &=-\frac{1+\alpha}{u}(1+O(u^{-\alpha}))\\
        &=-\frac{1+\alpha}u+O(u^{-1-\alpha}).
    \end{align*}
    For the derivative of the score function, first note that
    \begin{align*}
        \frac{p^{(k)}(u)}{p(u)}&=\frac{c_{k,1}}{c_{0,1}}\,u^{-k}\cdot\frac{1+r_k}{1+r_0}\\
        &=(-1)^k\prod_{\ell=1}^{k}(\ell+\alpha)\,u^{-k}(1+O(u^{-\alpha}))\\
        &=(-1)^k\prod_{\ell=1}^{k}(\ell+\alpha)\,u^{-k}+O(u^{-k-\alpha}).
    \end{align*}
    Consequently,
    \begin{align*}
        s_X'(u)&=\frac{p_X''(u)}{p_X(u)}-s_X^2(u)\\
        &=\frac{(2+\alpha)(1+\alpha)}{u^2}-\frac{(1+\alpha)^2}{u^2}+O(u^{-2-\alpha})\\
        &=\frac{1+\alpha}{u^2}+O(u^{-2-\alpha})
    \end{align*}
    and similarly
    \begin{align*}
        s_X''(u)&=\frac{p^{(3)}(u)}{p(u)}-3\frac{p''(u)}{p(u)}s_X(u)+2s_X^3(u)\\
        &=-\frac{2(1+\alpha)}{u^3}+O(u^{-3-\alpha}).
    \end{align*}
    Therefore, \(\zeta_X(u)=s_X(u)+\frac{1+\alpha}{u}=O(u^{-1-\alpha})\), \(\zeta_X'(u)=s_X'(u)-\frac{1+\alpha}{u^2}=O(u^{-2-\alpha})\), and \(\zeta_X''(u)=s_X''(u)+\frac{2(1+\alpha)}{u^3}=O(u^{-3-\alpha})\). The \(O\)-bounds hold for sufficiently large \(u> u_0\). Over the interval \([1,u_0]\), positivity and smoothness of \(p_X(\cdot)\) makes \(s_X(u)=\tfrac{p_X'(u)}{p_X(u)}\in C^\infty\). So, \(\zeta_X^{(k)}(u)\), for all $k \ge 0$, is continuous, and therefore bounded on $|u| \geq 1$. For \(u<-1\), it also holds true since \(\zeta_X^{(k)}(u)\) is either an odd or even function, so \(|\zeta_X^{(k)}(u)|=|\zeta_X^{(k)}(-u)|\). Finally,
    \begin{itemize}
        \item $|s_X(u)|\le \tfrac{C}{|u|}$ for $|u|\ge1$: from $|s_X|\le\frac{1+\alpha}{|u|}+C_X|u|^{-1-\alpha}\le\frac{1+\alpha+C_X}{|u|}$ on $|u|\ge1$.
        \item $\sup_u|s_X|<\infty$: bounded by the previous estimate for $|u|\ge1$, and by continuity on $[-1,1]$.
    \end{itemize}
\end{proof}

\section{Technical Lemmas}
\label{app:proofs}
\begin{definition}[Difference kernel]
Fix \(x_1\neq x_2\). Define
\begin{equation*}
\begin{aligned}
\Delta(y,x_1,x_2)
&:=p_{Y|X}(y|x_1)-p_{Y|X}(y|x_2)\\
&=p_Y(y)\!\left[
\frac{p_A(x_1-y)}{p_X(x_1)}
-\frac{p_A(x_2-y)}{p_X(x_2)}
\right].
\end{aligned}
\label{eq:Delta}
\end{equation*}
\end{definition}

% ============================================================
% Properties of the difference kernel
% ============================================================
\begin{property}[Properties of the difference kernel]
The function \(\Delta(y,x_1,x_2)\) satisfies
\begin{align}
\textnormal{(P1)}\quad&
\int_{\mathbb R}\Delta(y,x_1,x_2)\,dy=0,
\label{eq:P1}\\
\textnormal{(P2)}\quad&
\int_{\mathbb R}y\,\Delta(y,x_1,x_2)\,dy
=\rho^\alpha(x_1-x_2).
\label{eq:P2}
\end{align}
Moreover, for some
\(C=C(x_1,x_2,\alpha,\gamma_Y,\gamma_A)\),
\begin{equation}
\textnormal{(P3)}\qquad
|\Delta(y,x_1,x_2)|\le C(1+|y|)^{-2-2\alpha}.
\label{eq:P3}
\end{equation}
Consequently, as \(R\to\infty\),
\begin{align}
\int_{|y|>R}|\Delta(y,x_1,x_2)|\,dy
&=O(R^{-1-2\alpha}), \label{eq:tail0}\\
\int_{|y|>R}|y|\,|\Delta(y,x_1,x_2)|\,dy
&=O(R^{-2\alpha}). \label{eq:tail1}
\end{align}
\end{property}

\begin{proof}
Equation~\eqref{eq:P1} follows because
\(p_{Y|X}(\cdot|x_1)\) and \(p_{Y|X}(\cdot|x_2)\) are
probability densities. Moreover,
\(\mathbb E[Y|X=x]=\rho^\alpha x\)
\cite{Chouity2026Linear,11653965},
which gives~\eqref{eq:P2}.

For~\eqref{eq:P3}, Proposition~\ref{prop:score} gives
\(p_Y(y)=O(|y|^{-1-\alpha})\) and
\(p_A(x_i-y)=O(|y|^{-1-\alpha})\) as
\(|y|\to\infty\). Since both densities are bounded and
\(p_X(x_i)>0\), \eqref{eq:P3} follows.
Equations~\eqref{eq:tail0}--\eqref{eq:tail1} then follow by
direct integration.
\end{proof}

% ============================================================
% IB stationarity
% ============================================================

\begin{lemma}[Necessary condition for global optimality of the linear encoder]
\label{lem:stationarity}
Let \(M:= T - k Y = kA+N\) be the effective noise in the induced channel \(Y\to T\) for the optimal stable linear encoder. For \(\beta>0\), suppose that
\(p_{T|X}(t|x)\) achieves the optimum IB and satisfies
\eqref{eq:selfconsistent}. Then, for every \(x_1\neq x_2\),
the function \(\xi(t,x_1,x_2)\) in~\eqref{eq:Xi} is constant in \(t\),
and
\begin{equation*}
\partial_t\xi_{\mathrm{lin,opt}}(t,x_1,x_2) \equiv 0,
\end{equation*}
where 
\begin{equation}
\partial_t\xi_{\mathrm{lin}}(t,x_1,x_2)=s_N(t-kx_1)-s_N(t-kx_2) -\beta\!\int_{\mathbb R}
\Delta(y,x_1,x_2)\,s_M(t-ky)\,dy. \label{eq:Xiprime}
\end{equation}
\end{lemma}

\begin{proof}
Taking logarithms in~\eqref{eq:selfconsistent} and expanding
the divergence gives
\begin{align}
&\log p_{T|X}(t|x)\nonumber\\
&=\log p_T(t)
+\beta\!\int p_{Y|X}(y|x)
   \log p_{Y|T}(y|t)\,dy
\nonumber\\
&-\beta\!\int p_{Y|X}(y|x)
   \log p_{Y|X}(y|x)\,dy
-\log Z(x,\beta).
\label{eq:IB-log}
\end{align}
The last two terms are independent of \(t\).

The first integral is absolutely convergent since at \(|y|\to\infty\)
\[
\begin{aligned}
\log p_{Y|T}(y|t)
&=\log p_Y(y)+\log p_M(t-ky)-\log p_T(t)\\
&=O(\log|y|),
\end{aligned}
\]
while
\(p_{Y|X}(y|x)= \frac{p_A(x-y)p_Y(y)}{p_x(x)} = O(|y|^{-2-2\alpha})\).

Evaluating~\eqref{eq:IB-log} at \(x_1\) and \(x_2\) and
subtracting shows that \(\xi(t,x_1,x_2)\) given in equation~\eqref{eq:Xi} is constant in \(t\), since
\(p_{T|X}(t|x_i)=p_N(t-kx_i)\), for $i = 1,2$. Furthermore,
\(
\partial_t\log p_{Y|T}(y|t)
=s_M(t-ky) - s_T(t),
\)
where $s_M(\cdot)$ and $s_T(\cdot)$ are the score functions of $p_M$ and $p_T$ respectively. Differentiation under the integral sign follows by dominated
convergence because \(\sup_u|s_M(u)|<\infty\) and \(\Delta(y,x_1,x_2)\in L^1(\mathbb R)\), for fixed \(x_1\neq x_2\) by~\eqref{eq:P3}.
The term \(s_T(t)\int_{\mathbb{R}} \Delta(y,x_1,x_2) \, dy\) vanishes by~\eqref{eq:P1}, while
\(\partial_t\log p_N(t-kx_i)=s_N(t-kx_i)\).
This proves~\eqref{eq:Xiprime}.
\end{proof}

\begin{lemma}[Asymptotics under the best linear encoder]\label{lem:stationarity-asymptotics}

For the linear encoder $T = kX + N$ given in Theorem~\ref{thm:optimal_LSIB_scalar}, the function $\xi(t,x_1,x_2)$ given in equation~\eqref{eq:Xi} satisfies
    \begin{equation}
    \lim_{t\to\infty}t^{2}\partial_t\xi_{\mathrm{lin}}(t,x_1,x_2)
      %&=-(1+\alpha)k(x_1-x_2)+\beta(1+\alpha)k\rho^\alpha(x_1-x_2)\nonumber\\
      =(1+\alpha)k(x_1-x_2)\bigl[\beta\rho^\alpha-1\bigr].
    \end{equation}
\end{lemma}

\begin{proof}
By Proposition~\ref{prop:score},
\(s_N(u)=\zeta_N(u)-(1+\alpha)/u\). Hence
\begin{align}
&s_N(t-kx_1)-s_N(t-kx_2)
\nonumber\\
&=
\zeta_N(t-kx_1)-\zeta_N(t-kx_2)
-(1+\alpha)
\left(
\frac{1}{t-kx_1}-\frac{1}{t-kx_2}
\right)
\nonumber\\
&=
\zeta_N(t-kx_1)-\zeta_N(t-kx_2)
-\frac{(1+\alpha)k(x_1-x_2)}{t^2}
+O(t^{-3})\nonumber
\end{align}
as \(t\to\infty\). By the mean value theorem,
\[
\zeta_N(t-kx_1)-\zeta_N(t-kx_2)
=k(x_2-x_1)\zeta_N'(\eta_t)
\]
for some \(\eta_t\) between \(t-kx_1\) and \(t-kx_2\).
Hence, equation~\eqref{eq:Phi-bounds} gives
\(\zeta_N'(\eta_t)=O(t^{-2-\alpha})=o(t^{-2})\),
\begin{equation}
s_N(t-kx_1)-s_N(t-kx_2) =
-\frac{(1+\alpha)k(x_1-x_2)}{t^2}
+o(t^{-2}).
\label{eq:encoder-asymptotic}
\end{equation}

Next we consider the term \(\int_{\mathbb R}
\Delta(y,x_1,x_2) \, s_M(t-ky)\,dy\). To this end, choose \(\delta\in\left(\frac{1}{1+2\alpha},1\right)\). Let \(J(t):=t^\delta\) 
and define, within this proof,
\[
\begin{aligned}
A_t&:=\{|y|\le J(t)\}, \qquad
C_t:=\{|t-ky|\le 1\},\\
B_t&:=\{|y|>J(t)\}\setminus C_t.
\end{aligned}
\]
For all sufficiently large \(t\), these sets partition
\(\mathbb R\).

\paragraph*{Inner region $A_t$}
Since \(k\tfrac{J(t)}{t}\to0\), for sufficiently large \(t\) and
\(y\in A_t\), any point \(\eta_{y,t}\) between \(t\) and
\(t-ky\) satisfies
\(\tfrac12t\le\eta_{y,t}\le2t\). Thus
\[
s_M(t-ky)
=s_M(t)-ky\,s_M'(t)
+\frac{k^2y^2}{2}s_M''(\eta_{y,t}).
\]
Using~\eqref{eq:score-exp},
\[
s_M(t)=O(t^{-1}),\quad
s_M'(t)=\frac{1+\alpha}{t^2}+O(t^{-2-\alpha}),
\]
and
\(\sup_{y\in A_t}|s_M''(\eta_{y,t})|=O(t^{-3})\).
Therefore,
\begin{equation}
\int_{A_t}\Delta(y,x_1,x_2)\,s_M(t-ky)\,dy
=s_M(t)\int_{A_t}\Delta(y,x_1,x_2)\,dy - k s_M'(t)\int_{A_t}y\,\Delta(y,x_1,x_2)\,dy
+R_t,
\label{eq:inner-decomp}
\end{equation}
where
\[
|R_t|
\le Ct^{-3}\int_{A_t}y^2\,|\Delta(y,x_1,x_2)|\,dy.
\]

By~\eqref{eq:P1} and~\eqref{eq:tail0},
\[
\int_{A_t} \Delta(y,x_1,x_2)\,dy
=O(J(t)^{-1-2\alpha}),
\]
so the first term in~\eqref{eq:inner-decomp} is
\(O(t^{-1-\delta(1+2\alpha)})=o(t^{-2})\).

Likewise, \eqref{eq:P2} and~\eqref{eq:tail1} give
\[
\int_{A_t}y\,\Delta(y,x_1,x_2)\,dy
=\rho^\alpha(x_1-x_2)+O(t^{-2\delta\alpha}),
\]
and hence
\begin{align*}
&-k s_M'(t)\int_{A_t}y\,\Delta(y,x_1,x_2)\,dy
\nonumber\\
&\quad=
-\frac{(1+\alpha)k\rho^\alpha(x_1-x_2)}{t^2}
+O(t^{-2-\alpha})
+O(t^{-2-2\delta\alpha}).
\label{eq:inner-main}
\end{align*}

Finally, equation~\eqref{eq:P3} yields
\[
\int_{A_t}y^2|\,\Delta(y,x_1,x_2)|\,dy
=
\begin{cases}
O(1), & \alpha>\frac12,\\
O(\log t), & \alpha=\frac12,\\
O(t^{\delta(1-2\alpha)}), & \alpha<\frac12.
\end{cases}
\]
Thus \(R_t=o(t^{-2})\), and
\begin{equation*}
\int_{A_t} \Delta(y,x_1,x_2)\, s_M(t-ky)\,dy
=
-\frac{(1+\alpha)k\rho^\alpha(x_1-x_2)}{t^2}
+o(t^{-2}).
\label{eq:inner-result}
\end{equation*}

\paragraph*{Outer region $B_t \cup C_t$}
Since \(k>0\) WLoG, We obtain
\[
C_t=\left[\frac{t-1}{k},\frac{t+1}{k}\right],
\qquad |C_t|=\frac{2}{k}.
\]

On \(C_t\), \(y=\tfrac{t}{k}+O(1)\), so
\(|\Delta(y,x_1,x_2)|=O(t^{-2-2\alpha})\).  Since \(\delta<1\), \(C_t\subset\{|y|>J(t)\}\) for large \(t\), and 
\[
\int_{C_t} \Delta(y,x_1,x_2)\,s_M(t-ky)\,dy
=O(t^{-2-2\alpha})=o(t^{-2}),
\]
where we  used Proposition~\ref{prop:score}
\[
|s_M(u)|\le \frac{C}{|u|},\quad |u|\ge1,
\qquad
\sup_u|s_M(u)|<\infty.
\]
On \(B_t\cap\{|y|\le \tfrac{t}{2k}\}\),
\(|t-ky|\ge t/2\), and therefore
\begin{align*}
&\left|
\int_{\substack{B_t\\ |y|\le t/(2k)}}
\Delta(y,x_1,x_2)\,s_M(t-ky)\,dy
\right|
\nonumber\\
&\quad\le
\frac{2C}{t}\int_{|y|>J(t)}|\Delta(y,x_1,x_2)|\,dy
=O(t^{-1-\delta(1+2\alpha)})
=o(t^{-2}).
\label{eq:outer-near}
\end{align*}

Finally, on \(B_t\cap\{\tfrac{t}{2k}\}\),
\(|t-ky|>1\). With \(w=ky-t\), the contribution is bounded by
\begin{equation*}
\frac{C}{k}
\int_{\substack{|w|>1\\ |w+t|>t/2}}
\frac{
\left|\Delta\!\left((w+t,x_1,x_2)/k\right)\right|
}{|w|}\,dw .
\label{eq:outer-w}
\end{equation*}
For \(1<|w|\le t/2\),
\(|w+t|\ge t/2\), so this part is
\[
O(t^{-2-2\alpha}\log t)=o(t^{-2}).
\]
For \(|w|>\tfrac{t}{2}\), the condition \(|w+t|>\tfrac{t}{2}\) implies either
\(w>\tfrac{t}{2}\) or \(w<-3\tfrac{t}{2}\), and hence
\(|w+t|\ge\frac{|w|}{3}\). Thus the remaining part is
\[
O\!\left(\int_{|w|>t/2}|w|^{-3-2\alpha}\,dw\right)
=O(t^{-2-2\alpha})
=o(t^{-2}).
\]
Combining the inner and outer contributions, we obtain
\begin{equation}
    \int_{\mathbb R}\Delta(y,x_1,x_2)\,s_M(t-ky)\,dy =
-\frac{(1+\alpha)k\rho^\alpha(x_1-x_2)}{t^2}
+o(t^{-2}).
\label{eq:decoder-asymptotic}
\end{equation}
Since
\begin{equation}
\partial_t\xi_{\mathrm{lin}}(t,x_1,x_2)
= s_N(t-kx_1)-s_N(t-kx_2) - \beta\!\int_{\mathbb R}
\Delta(y,x_1,x_2)\,s_M(t-ky)\,dy,\nonumber
\end{equation}
using \eqref{eq:encoder-asymptotic} and \eqref{eq:decoder-asymptotic} gives
\begin{align}
t^2\,\partial_t\xi_{\mathrm{lin}}(t,x_1,x_2)&=-(1+\alpha)k(x_1-x_2) + \beta(1+\alpha)k\rho^\alpha(x_1-x_2)+o(1)\nonumber\\
    &=(1+\alpha)k(x_1-x_2)[\beta\rho^\alpha-1]+o(1)\nonumber,
\end{align}
which completes the proof.
\end{proof}

\section*{Acknowledgment}
The research of M. Assaad was funded in part by the Agence Nationale
de la Recherche (ANR) under the ANR-24-IAS1-0002-02
COMSEMA project.  The authors acknowledge the use of ChatGPT for LaTeX-related assistance and code synthesis used to produce the graphs in Figure~\ref{fig:three_cases}. The use of Claude was limited to assistance in the evaluation of the integrals in the proof of Lemma~\ref{lem:stationarity-asymptotics} and the formulation of high-level arguments in Appendix~\ref{app:alpha-stable_probability_laws}, which were also independently verified by the authors. The main ideas in this paper, the technical content and most of the derivations, along with the results and conclusions, were produced by the authors.  

\bibliographystyle{IEEEtran}
\bibliography{ref}

\end{document}